\documentclass{amsart}

\usepackage{amsmath, amssymb, amsfonts, mathtools}
\usepackage{amsthm}
\usepackage{amsrefs}
\usepackage{enumitem}
\usepackage{accents}
\usepackage{mathrsfs}
\usepackage[a4paper,margin=1.32in]{geometry}
\mathtoolsset{showonlyrefs=true}  
\usepackage{hyperref}
\hypersetup{
    colorlinks = true,
    linkcolor  = DutchOrange,    
    citecolor  = DutchOrange,    
    urlcolor   = DutchOrange,    
}

\newtheorem{theorem}{Theorem}
\newtheorem{corollary}{Corollary}
\newtheorem{remark}{Remark}

\newtheorem{assumption}{Assumption}
\newtheorem{definition}{Definition}

\usepackage{tikz,pgfplots,pgfplotstable}
\usepackage{cuted} 

\usepackage{tikz-cd}
\usepgfplotslibrary{groupplots}
\usetikzlibrary{arrows.meta}
\usetikzlibrary{patterns.meta} 
\usepackage{graphicx}  
\usepackage{subcaption}
\pgfplotsset{compat=newest}

\usepackage[ruled,vlined]{algorithm2e}

\usepackage[svgnames]{xcolor} 
\DefineNamedColor{named}{LUblue}{cmyk}{1,0.85,0.05,0.22}
\DefineNamedColor{named}{LUbronze}{cmyk}{0.09,0.57,1,0.41}
\DefineNamedColor{named}{LuPastelle}{cmyk}{0,0.15,0.05,0.0}
\DefineNamedColor{named}{LuPastelle2}{cmyk}{0.29,0.02,0.24,0.03}
\DefineNamedColor{named}{LuPastelle3}{cmyk}{0.24,0.03,0.07,0.02}
\DefineNamedColor{named}{gentledark}{gray}{.15}
\colorlet{LUdarkblue}{LUblue!60!black}
\DefineNamedColor{named}{DutchOrange}{rgb}{1,0.498,0} 
\DefineNamedColor{named}{DutchBlue}{rgb}{0,0.239,0.647} 

\usepackage{microtype}

\usepackage{cite}

\usepackage[acronym]{glossaries}
\newacronym{mimo}{MIMO}{multi-input multi-output}
\newacronym{lti}{LTI}{linear-time-invariant}
\newacronym[plural=LMIs,
             longplural=linear matrix inequalities]{lmi}{LMI}{linear matrix inequality}
\newacronym[plural=SRGs,
             firstplural=Scaled Relative Graphs]{srg}{SRG}{Scaled Relative Graph}

\def\BibTeX{{\rm B\kern-.05em{\sc i\kern-.025em b}\kern-.08em
    T\kern-.1667em\lower.7ex\hbox{E}\kern-.125emX}}

\usepackage{etoolbox}
\makeatletter
\patchcmd{\equation}{\@beginparpenalty}{\@lowpenalty}{}{}
\patchcmd{\endequation}{\@endparpenalty}{\@lowpenalty}{}{}
\makeatother

\makeatletter

\def\HyperAbsTol{0.5pt}
\def\HyperRelTol{1e-6}

\def\hyper@x#1,#2\relax{#1} 
\def\hyper@y#1,#2\relax{#2} 
\def\hyper@coords#1{#1}     

\pgfmathdeclarefunction{safeDiv}{2}{%
  \pgfmathparse{ifthenelse(abs(#2) < \hyper@tol, 0, #1/#2)}%
}

\def\hyper@computer#1#2{%
  \edef\hyper@toscan{(#1)}%
  \tikz@scan@one@point\hyper@coords\hyper@toscan
  \edef\hyper@sx{\the\pgf@x}%
  \edef\hyper@sy{\the\pgf@y}%
  \edef\hyper@toscan{(#2)}%
  \tikz@scan@one@point\hyper@coords\hyper@toscan
  \edef\hyper@ex{\the\pgf@x}%
  \edef\hyper@ey{\the\pgf@y}%

  \pgfmathsetmacro{\hyper@mx}{(\hyper@ex + \hyper@sx)/2}%
  \pgfmathsetmacro{\hyper@my}{(\hyper@ey + \hyper@sy)/2}%
  \pgfmathsetmacro{\hyper@dy}{\hyper@ey - \hyper@sy}%
  \pgfmathsetmacro{\hyper@dx}{\hyper@ex - \hyper@sx}%

  \pgfmathsetmacro{\hyper@L}{veclen(\hyper@dx,\hyper@dy)}%
  \pgfmathsetmacro{\hyper@tol}{max(\HyperAbsTol, \HyperRelTol*\hyper@L)}%

  \pgfmathtruncatemacro{\isvertical}{abs(\hyper@dx) < \hyper@tol}%
  \pgfmathtruncatemacro{\ishorizontal}{abs(\hyper@dy) < \hyper@tol && (abs(\hyper@dx) < 0.5)}%

  \ifnum\numexpr\isvertical+\ishorizontal>0\relax
    \edef\hyper@cmd{-- (\tikztotarget)}%
  \else
    \pgfmathsetmacro{\hyper@t}{safeDiv(\hyper@my, \hyper@dx)}%
    \pgfmathsetmacro{\hyper@cx}{\hyper@mx + \hyper@t * \hyper@dy}%
    \pgfmathsetmacro{\hyper@radius}{veclen(\hyper@cx - \hyper@sx, \hyper@sy)}%
    \pgfmathsetmacro{\hyper@sangle}{180 - atan2(\hyper@sy,\hyper@cx-\hyper@sx)}%
    \pgfmathsetmacro{\hyper@eangle}{180 - atan2(\hyper@ey,\hyper@cx-\hyper@ex)}%
    \edef\hyper@cmd{arc[radius=\hyper@radius pt, start angle=\hyper@sangle, end angle=\hyper@eangle]}%
  \fi
}

\tikzset{%
  hyperbolic disc radius/.initial={1cm},
  hyperbolic plane/.style={
    to path={
      \pgfextra{\hyper@computer\tikztostart\tikztotarget}%
      \hyper@cmd
    }
  },
  hyperbolic disc/.style={
    to path={
      \pgfextra{\hyper@disc@computer\tikztostart\tikztotarget}%
      \hyper@cmd
    }
  },
  hyperbolic plane target angle/.initial={},
}

\pgfmathdeclarefunction{bkx}{2}{%
\begingroup
 \pgfmathparse{1-(2/(1+#1*#1+#2*#2))}
 \pgfmath@smuggleone\pgfmathresult%
\endgroup}

\pgfmathdeclarefunction{bky}{2}{%
\begingroup
 \pgfmathparse{-(2*#1/(1+#1*#1+#2*#2))}
 \pgfmath@smuggleone\pgfmathresult%
\endgroup}

\makeatother

\newcommand{\x}{x}  
\newcommand{\y}{y}  
\renewcommand{\u}{u}  
\renewcommand{\k}{k} 
\newcommand{\T}{\ensuremath{\mathbf{L}}}  
\newcommand{\A}{A}  
\newcommand{\B}{B}  
\newcommand{\C}{C}  
\newcommand{\D}{D}  
\renewcommand{\P}{P}  

\renewcommand{\Re}{\operatorname{Re}}
\renewcommand{\Im}{\operatorname{Im}}

\newcommand{\hilbert}{\mathcal{H}}

\newcommand{\Real}{\ensuremath{\mathbb{R}}}

\newcommand{\Compex}{\ensuremath{\hat{\mathbb{C}}}}
\newcommand{\cl}[1]{\ensuremath{\mathrm{cl}\,{#1}}}
\newcommand{\id}{\ensuremath{\mathbf{I}}}

\newcommand{\srg}[1]{\ensuremath{\mathrm{SRG}\left(#1\right)}} 
\newcommand{\srgc}[1]{\ensuremath{\mathrm{cl}\;\mathrm{SRG}(#1)}} 

\newcommand{\norm}[1]{\left\lVert#1\right\rVert}
\newcommand{\dom}[1]{\ensuremath{\mathcal{D}_{#1}}} 

\newcommand{\maxgain}[1]{\ensuremath{\bar{\sigma}\left(#1\right)}}
\newcommand{\mingain}[1]{\ensuremath{\ubar{\sigma}\left(#1\right)}}
\newcommand{\elltwotau}{\ensuremath{\ell_{2,\tau}^{m}}}
\newcommand{\elltwoinf}{\ensuremath{\ell_{2}^{m}}} 
 
\newcommand{\Tinf}{\ensuremath{\mathbf{T}}}
\newcommand{\Ttau}{\ensuremath{\mathbf{T}^\tau}}

\newcommand{\ubar}[1]{\underaccent{\bar}{#1}}

\begin{document}

\title{Computing Scaled Relative Graphs of Discrete-time LTI Systems from Data}

\author{Talitha Nauta and Richard Pates}
\thanks{The authors are with the Department of Automatic Control, Lund University, Box 118, SE 221 00 LUND, Sweden (e-mail: {\tt\{talitha.nauta, richard.pates\}@control.lth.se}).}
\thanks{The authors are members of the ELLIIT
Strategic Research Area at Lund University. This work was supported by the ELLIIT Strategic Research Area.}

\maketitle

\begin{abstract}
    Graphical methods for system analysis have played a central role in control theory. The \gls{srg} has recently emerged as a useful tool for stability analysis of feedback interconnections. In this paper, we further extend its applicability by showing how the \gls{srg} of a discrete-time \gls{lti} system can be computed exactly from its state-space representation using linear matrix inequalities. We additionally propose a fully data-driven approach where we demonstrate how to compute the \gls{srg} exclusively from input-output data. Furthermore, we introduce a robust version of the \gls{srg}, which can be computed from noisy data trajectories and contains the \gls{srg} of the actual system.
    \glsresetall
\end{abstract}

\section{Introduction}

Graphical system analysis with tools such as Nyquist and Bode diagrams has been foundational for the development of classical control theory. These tools can be used for stability and robustness analysis, as well as loop shaping and system identification. 
Lately, the \gls{srg} has emerged as a new promising graphical tool for the analysis of dynamical systems. Originally introduced by Ryu, Hannah, and Yin in \cite{Ryu2021} to establish rigorous proofs for convergence of optimisation algorithms, the framework was subsequently connected to classical systems theory in~\cite{Chaffey2021,Chaffey2023}. In this field, it has been demonstrated that the \gls{srg} can be used in several areas of control theory, such as stability and robustness analysis of both linear and non-linear \gls{mimo} systems~\cite{Chaffey2021,Chaffey2023,baronprada2025,baronprada2025a,chen2025a,krebbekx2025}. Other examples of applications are in connection with integral quadratic constraints and dissipativity~\cite{chen2025a,degroot2025,Eijnden2024}, model reduction~\cite{Chaffey2022}, and the Lur'e problem~\cite{krebbekx2025a}.

To fully exploit the capabilities of the \gls{srg}, it is crucial to understand how it can be determined. Recently, significant efforts have been made to develop methods for computing, or over-approximating, the \glspl{srg} of operators associated with dynamical systems. The case was solved for bounded linear operators on Hilbert spaces in~\cite{Pates21}, where it was shown that the \gls{srg} can be obtained through a specific mapping of the Numerical Range. 
For stable \gls{lti} systems, initial results were provided in~\cite{Chaffey2021,Chaffey2023}, where \glspl{srg} for single-input single-output systems were obtained using Nyquist-like diagrams.
These results were then generalised to normal \gls{lti} systems via dissipativity results in~\cite{degroot2025}, which also presents techniques for over-approximating \glspl{srg} of \gls{mimo} \gls{lti} systems and certain non-linearities.

General dynamical systems are often modelled as operators over extended spaces. Since these spaces are not Hilbert spaces, constructing \glspl{srg} for dynamical systems is challenging, because the \glspl{srg} are defined specifically over Hilbert spaces. To bridge this gap, the notion of soft and hard \glspl{srg} was introduced in~\cite{chen2025a}, enhancing compatibility with results from integral quadratic constraints and dissipativity. This led to results for over-approximation of \glspl{srg}~\cite{degroot2025,krebbekx2025} and for obtaining hard \glspl{srg}~\cite{krebbekx2025b}. Precise construction of \glspl{srg} for closed linear operators was later shown in~\cite{Nauta25}, where maximum and minimum gain calculations are used to compute the \gls{srg}. This includes operators that are closely related to the notion of soft and hard \glspl{srg}.

Existing results on computing \glspl{srg} for \gls{lti} systems are for continuous-time. In this paper, we extend these results to exact characterisations of \glspl{srg} for discrete-time \gls{lti} \gls{mimo} systems on state-space form. This allows us to further develop the results to the data-driven setting, which we do by connecting the characterisation methods to dissipativity properties, as first proposed in~\cite{degroot2025}. We exploit this connection by building on the results of~\cite{Koch22}, where it is shown that dissipativity properties can be verified using data trajectories.
Our contribution is threefold. First, we show in Section~\ref{sec:statespace} how to compute the \gls{srg} of a discrete-time \gls{lti} system using \glspl{lmi} based on the state-space representation. Secondly, we show how to use \glspl{lmi} based on data trajectories to compute the \gls{srg} of an unknown system in Section~\ref{sec:trajectories}. Lastly, in Section~\ref{sec:noisetrajectories}, we develop these results into robust over-approximations of the \gls{srg} from noisy data trajectories. Examples are shown in Section~\ref{sec:examples}.

\section{Preliminaries}
\label{sec:preliminaries}
\subsection{Notation}
We denote the closure of a set $S$ as $\cl{S}$. The extended complex plane is defined as $\Compex\coloneqq{}\mathbb{C} \cup \{\infty\}$. For a sequence of subsets $S_k\subseteq{}\Compex$, we define the limit as $\lim_{k\rightarrow{}\infty}S_k\coloneqq{}\{z\in\Compex:z=\lim_{k\rightarrow{}\infty}z_k,z_k\in S_k\}.$ We let $(*)^\top \Pi M$ denote $M^\top \Pi M$, where $(*)$ inherits the block structure of~$M$.

$\hilbert$ denotes a Hilbert space over the field $\mathbb{C}$. The space is equipped with an inner product $\langle\cdot,\cdot\rangle: \hilbert~\times~\hilbert~\rightarrow~\mathbb{C}$, which induces a norm $\norm{\cdot} = \sqrt{\langle \cdot,\cdot \rangle}$. The Hilbert space of square-summable sequences over $\mathbb{C}$ is the set of sequences $$\elltwoinf \coloneqq \{ f:  (0, 1,2,\dots) \rightarrow \mathbb{C}^m : \norm{f} < \infty \},$$ with the inner product $\langle u, y \rangle = \sum_{k=0}^\infty \u(k)^* \y(k)$. The Hilbert space $\elltwotau$ is the linear subspace of $\elltwoinf$, where $\tau \in \{1,2,3,\dots\}$ and $f(k) = 0$ when $k \in \{\tau, \tau +1, \tau +2, \dots\}$. 

We call $\T:\dom{\T} \subseteq \hilbert \rightarrow \hilbert$ a linear, possibly unbounded, operator if it satisfies linearity and $\dom{\T}$ is a linear manifold. The operator is closed if its graph is a closed subset of $\hilbert \times \hilbert$. We denote the identity operator by $\id$. The maximum and minimum gain of an operator $\T$ we define by 
\begin{align}
      \maxgain{\T}
   \coloneqq{}
   \sup_{\substack{\u \in \dom{\T} \\ \u \neq 0}} \frac{\norm{\T\u}}{\norm{\u}}
 \;\;\text{  and  }\;\;
      \mingain{\T}
   \coloneqq{}
   \inf_{\substack{\u \in \dom{\T} \\ \u \neq 0}}  \frac{\norm{\T \u}}{\norm{\u}}. 
\end{align}

\subsection{Scaled Relative Graphs}
The \gls{srg} is a subset of the extended complex plane. For a possibly unbounded linear operator $\T:\dom{\T} \subseteq \hilbert \rightarrow \hilbert$ we define the \gls{srg} as follows:
\begin{align}
    \srg{\T} \!
    \coloneqq{} \!\!
    \left\{ \frac{\| \y \|}{\| \u\|} \text{exp} \left(\! \pm i \angle(\u,\y) \right) \! : \!
    u \in \dom{\T} \backslash \{ 0 \} , \y = \T \u \right\}\!,
\end{align}
where the angle $\angle(u,y)$ between $u \in \mathcal{H}$ and $y \in \mathcal{H}$ is defined by its inner product
\begin{align}
    \cos{(\angle(u,y))} = \frac{\text{Re}(\langle y,u \rangle)}{\| y \| \| u \|} \text{ where } \angle(u,y) \coloneqq 0 \text{ if } y = 0.
\end{align}
This naturally extends the usual definition in \cite{Ryu2021} to cover operators whose domain is not necessarily equal to $\hilbert$. 
Note that since the operators we consider are linear, the \textit{relative} part of the definition is not required.
The \gls{srg} characterises certain geometric properties of the operator $\T$. 
The term $\|y\|/\|u\|$ quantifies the gain, and the exponent captures the phase shift between the input and output~\cite{Chaffey2021,Chaffey2023}.

\section{SRGs of discrete-time state-space models}
\label{sec:statespace}
In the following section, we show how to compute \glspl{srg} of operators associated with discrete-time \gls{lti} \gls{mimo} systems. We consider a system for which there exists a minimal (controllable and observable) realisation of the form
\begin{equation}
\label{eq:system}
\begin{aligned}
  \x_{\k+1} &= \A\x_{\k} + \B\u_{\k}, \,\, \x_0 = 0, \\
  \y_{\k} &= \C\x_{\k} + \D\u_{\k},
\end{aligned}
\end{equation}
where the state $\x_{\k} \in \Real^{n}$, input $\u_{\k} \in \Real^{m}$, and the output $\y_{\k} \in \Real^{m}$. Note that we require the number of inputs and outputs to be equal. We define two types of operators from $\u$ to $\y$ associated with this system.
\begin{definition}
\label{def:Ttau}
    Given a $\tau \in \{1,2,3,\dots\}$ we define the operator
    $\Ttau: \elltwotau \rightarrow \elltwotau$,
   implicitly so that the input maps to the output as follows 
\begin{equation}
\begin{gathered}
    (u_0, u_1, \cdots, u_\tau, 0, \cdots) \mapsto  \left(Du_0, CBu_0\!+\!Du_1, \! \ldots \!, \Sigma_{i=0}^{\tau-1} CA^{\tau-i-1}Bu_i + Du_\tau, 0,\!\ldots \right) \! .
    \end{gathered}
\end{equation}
Note that both the in- and outputs are truncated so that $u_i=0$ and $y_i=0$ for $i > \tau$, which ensures they lie in $\elltwotau$. 
\end{definition}

\begin{definition}
\label{def:Tinf}
    We define the operator 
    $\Tinf: \dom{\Tinf} \subseteq \elltwoinf \rightarrow \elltwoinf$,
    implicitly by the map
    \begin{equation}
\begin{gathered}
    (u_0, u_1, \cdots, u_j, \ldots) \mapsto  \left(Du_0, CBu_0 + Du_1, \cdots, \Sigma_{i=0}^{j-1} CA^{j-i-1}Bu_i + Du_j, \ldots \right).
    \end{gathered}
\end{equation}
    where the domain is restricted such that all outputs lie in $\elltwoinf$, meaning that $\dom{\Tinf} \coloneqq{} \{\u\in\elltwoinf:\y\in\elltwoinf\}$. 
\end{definition}
From~\cite{chen2025a} it follows that both $\lim_{\tau \rightarrow \infty} \srg{\Ttau}$ and $\srg{\Tinf}$ are interesting for stability and robustness analysis of the system.
It was shown in~\cite{Nauta25} that the \gls{srg} of any closed linear operator can be computed through maximum and minimum gain calculations. As both $\Ttau$ and $\Tinf$ are closed, we can use the following result based on \cite[Theorem 2]{Nauta25}. 
\begin{corollary}
\label{cor:circle_thm}
    Consider the operators $\Ttau$ and $\Tinf$ as defined in Definition \ref{def:Ttau} and \ref{def:Tinf}. Then  $\lim_{\tau \rightarrow \infty} \srg{\Ttau}$ is equal to
    \begin{equation}
         \bigcap_{\alpha \in \Real} \left\{ \alpha + z : \!\lim_{\tau \rightarrow \infty}  \mingain{\Ttau - \! \alpha \id} \leq |z| \!  \leq \!  \lim_{\tau \rightarrow \infty}  \maxgain{\Ttau - \! \alpha \id} \right\},
    \end{equation}
    and $\srgc{\Tinf}$ is equal to
    \begin{equation}
         \bigcap_{\alpha \in \Real} \left\{ \alpha + z :\mingain{\Tinf - \alpha \id} \leq |z| \leq \maxgain{\Tinf - \alpha \id} \right\}.
    \end{equation}
\end{corollary}
This means that $\lim_{\tau \rightarrow \infty} \srg{\Ttau}$ and $\srg{\Tinf}$ can be calculated to arbitrary precision by computing the maximum and minimum gain over a grid of $\alpha$, see Figure \ref{fig:bounding_ann}. Such an $\alpha$ grid could, for example, be equally spaced points centred around zero. The computation of gains for $\Tinf$ and $\lim_{\tau \rightarrow \infty} \Ttau$ can be done using the state-space formulas below. 
\begin{figure}
    \centering
      \begin{tikzpicture}[every to/.style={hyperbolic plane},scale=0.61,>={Stealth[scale=1]}]  
       \pgfmathsetmacro{\myalphaone}{-2.5}  
       \pgfmathsetmacro{\maxrone}{3.2566}
       \pgfmathsetmacro{\minrone}{1}
       \pgfmathsetmacro{\myalphatwo}{-1.3}  
       \pgfmathsetmacro{\maxrtwo}{2.0999}
       \pgfmathsetmacro{\minrtwo}{0.5262}
       \pgfmathsetmacro{\myalphathree}{1.5}  
       \pgfmathsetmacro{\maxrthree}{3.3423}
       \pgfmathsetmacro{\minrthree}{0.8740}

        \let\radius\undefined
        \newlength{\radius}
        \setlength{\radius}{1.5pt}

        \begin{scope}
          \clip (\myalphaone,0) circle (\maxrone);
          \clip (\myalphatwo,0) circle (\maxrtwo);
          \clip (\myalphathree,0) circle (\maxrthree);

          \fill[black!20] (-5,-5) rectangle (5,5);

          \fill[white] (\myalphaone,0) circle (\minrone);
          \fill[white] (\myalphatwo,0) circle (\minrtwo);
          \fill[white] (\myalphathree,0) circle (\minrthree);

        \end{scope}

         \draw[thick, green, even odd rule] (\myalphaone,0) circle (\maxrone) (\myalphaone,0) circle (\minrone);
        \draw[thick, blue, even odd rule] (\myalphatwo,0) circle (\maxrtwo) (\myalphatwo,0) circle (\minrtwo);
        \draw[thick, red, even odd rule] (\myalphathree,0) circle (\maxrthree) (\myalphathree,0) circle (\minrthree);

        \input{data/SRG/corollary_example2.tex}

        \fill (\myalphaone,0) circle[radius=\radius];
        \fill (\myalphatwo,0) circle[radius=\radius];
        \fill (\myalphathree,0) circle[radius=\radius];
        \node[below, green] at (\myalphaone,0) {$\alpha_1$};
        \node[below, blue] at (\myalphatwo,0) {$\alpha_2$};
        \node[below, red] at (\myalphathree,0) {$\alpha_3$};
        
        \draw[->] (-6.5,0) -- (6,0) node[below]{$\Re$};
        \draw[->] (0,-3.4) -- (0,3.5) node[left]{$\Im$};
        \draw (\myalphaone,.05) -- (\myalphaone,-.15);
        \draw (\myalphatwo,.05) -- (\myalphatwo,-.15);
        \draw (\myalphathree,.05) -- (\myalphathree,-.15);
    \end{tikzpicture}
    \caption{The figure shows a \gls{srg} (orange) and its approximation (grey) obtained by the gain bounds for $\{ \alpha_1, \alpha_2, \alpha_3\}$ computed following Corollary~\ref{cor:circle_thm}.
    }
    \label{fig:bounding_ann}
\end{figure}

\begin{theorem}
    Consider the operators $\Ttau$ and $\Tinf$ as defined in Definition \ref{def:Ttau} and \ref{def:Tinf}, coming from a minimal state-space realisation as in~\eqref{eq:system}. Then 
    \begin{align}    
    & \lim_{\tau \rightarrow \infty} \maxgain{\Ttau} = \inf \gamma \\
    & \text{such that there exists a}\;\; \P = \P^\top \succeq 0 \;\; \text{satisfying}: \\
    & \begin{bmatrix}
        \A^\top \P \A -\P + \C^\top \C & \A^\top \P \B +  \C^\top \D   \\
        \B^\top \P \A + \D^\top \C & \B^\top \P \B + \D^\top \D - \gamma^2 I
    \end{bmatrix}  \preceq 0
    \label{eq:LMI_maxgain}
    \end{align}
    and 
    \begin{align}    
    & \lim_{\tau \rightarrow \infty} \mingain{\Ttau} = \sup \zeta \\
    & \text{such that there exists a}\;\; \P = \P^\top \succeq 0 \;\; \text{satisfying}: \\
    & \begin{bmatrix}
        \A^\top \P \A -\P - \C^\top \C & \A^\top \P \B -  \C^\top \D   \\
        \B^\top \P \A - \D^\top \C & \B^\top \P \B - \D^\top \D + \zeta^2 I
    \end{bmatrix}  \preceq 0.
    \label{eq:LMI_mingain}
    \end{align}
    If we remove the constraint $\P \succeq 0$ from \eqref{eq:LMI_maxgain} and \eqref{eq:LMI_mingain}, the \glspl{lmi} give $\maxgain{\Tinf} = \inf \gamma$ and $\mingain{\Tinf} = \sup \zeta$ instead.  
\label{thm:LMIs}
\end{theorem}
\begin{remark}
    For $\Tinf - \alpha \id$ and $\lim_{\tau \rightarrow \infty} \Ttau - \alpha \id$ the gains can be obtained by substituting $D$ with $D - \alpha I$ in~\eqref{eq:LMI_maxgain} and~\eqref{eq:LMI_mingain}.
\end{remark}
\begin{proof} 
    The formula for $\lim_{\tau \rightarrow \infty} \maxgain{\Ttau}$ follows directly from the Bounded Real Lemma. For $\lim_{\tau \rightarrow \infty} \mingain{\Ttau}$, we will show that we can rewrite the problem as a maximum gain problem for the inverted system. The operator $\Ttau$ is invertible if and only if $D$ is non-singular. If $D$ is singular we can take an input trajectory $(0, 0, \cdots, 0, u_\tau, 0, \cdots)$ such that $u_\tau \neq 0$ but $Du_\tau = 0$, which means $y_k = 0$ for all $k$, and the minimum gain is 0. In this case, \eqref{eq:LMI_mingain} is only feasible for $\zeta = 0$, as $D^\top D$ has at least one zero eigenvalue. If $D$ is non-singular, the inverse system is given by 
    \begin{equation}
    \label{eq:inv_system}
    \begin{aligned}
      \x_{\k+1} &= (\A - \B\D^{-1}\C)\x_{\k} + \B\D^{-1}\y_{\k} = \Bar{A}\x_{\k} + \Bar{B}\y_{\k} \\
      \u_{\k} &= -\D^{-1}\C\x_{\k} + \D^{-1}\y_{\k} = \Bar{C}\x_{\k} + \Bar{D}\y_{\k}, \,\, \x_0 = 0.
    \end{aligned}
    \end{equation}
    For any invertible operator, $\mingain{\T} = 1/\maxgain{\T^{-1}}$ and therefore, $$\lim_{\tau \rightarrow \infty} \mingain{\Ttau} = \lim_{\tau \rightarrow \infty} 1/\maxgain{(\Ttau)^{-1}}.$$ We will now show that the \gls{lmi} in \eqref{eq:LMI_mingain} is equivalent to the \gls{lmi} in \eqref{eq:LMI_maxgain} when the state-space matrices are substituted with $\Bar{A}, \Bar{B}, \Bar{C}$ and $\Bar{D}$ and $\gamma = 1/\zeta$. The \gls{lmi} in  \eqref{eq:LMI_mingain} can be rewritten as
    \begin{equation}
    \begin{aligned}
    \setlength\arraycolsep{4pt} 
    \underbrace{\left(
    \begin{array}{cc}
    \star & \star \\ 
    \star & \star \\ 
    \star & \star \\ 
    \star & \star
    \end{array}
    \right)^\top}_{M^\top} 
    \underbrace{\left(
    \begin{array}{cccc}
    P&0&0&0\\
    0&-I&0&0\\
    0&0&-P&0\\
    0&0&0&\zeta^2I
    \end{array}\right)}_{\Pi}
     \underbrace{\left( \!\!
    \begin{array}{cc}
    A & B \\
    C & D \\ 
    I & 0 \\ 
    0 & I
    \end{array} \!\!
    \right)}_{M} \preceq 0.
    \end{aligned}
    \label{eq:LMI_big}
    \end{equation}
    We define the matrices
    \begin{equation}
        Q \coloneqq \begin{pmatrix}
            I & 0 & 0 & 0 \\
            0 & 0 & 0 & I \\
            0 & 0 & I & 0 \\
            0 & I & 0 & 0 \\
        \end{pmatrix} \quad Z \coloneqq \begin{pmatrix}
            I & 0 \\
            C & D
        \end{pmatrix}^{-1}.
    \end{equation}
    As $Q^2 = (Q^\top)^2 = I$ and $Z$ is invertible, the \gls{lmi} in \eqref{eq:LMI_big} is equivalent to
    $$Z^\top M^\top(Q^\top)^2 \Pi Q^2 M Z \preceq 0,$$
    which gives 
    \begin{equation}
    \begin{aligned}
    \setlength\arraycolsep{4pt} 
    \left(
    \begin{array}{cc}
    \star & \star \\ 
    \star & \star \\ 
    \star & \star \\ 
    \star & \star
    \end{array}
    \right)^\top 
    \begin{pmatrix}
    \widetilde{P}&0&0&0\\
    0&I&0&0\\
    0&0&-\widetilde{P}&0\\
    0&0&0&-\gamma^2I
    \end{pmatrix}
     \left( \!\!
    \begin{array}{cc}
    \Bar{A} & \Bar{B} \\
    \Bar{C} & \Bar{D} \\ 
    I & 0 \\ 
    0 & I
    \end{array} \!\!
    \right) \preceq 0,
    \end{aligned}
    \label{eq:LMI_inv}
    \end{equation}
    where $\gamma = \frac{1}{\zeta}$ and $\widetilde{P} = \frac{1}{\zeta^2}P$. Then the Bounded Real Lemma gives that the smallest $\gamma$ to satisfy \eqref{eq:LMI_inv} is $\lim_{\tau \rightarrow \infty} 1/\maxgain{(\Ttau)^{-1}}$. Therefore, the maximum $\zeta$ to satisfy \eqref{eq:LMI_mingain} gives $\lim_{\tau \rightarrow \infty} \mingain{\Ttau}$. 
    If we remove the constraint $P \succeq 0$, the result for both the maximum and minimum gain follow from the KYP lemma, as the state-space realisation is assumed to be controllable~\cite[Theorem~1.4]{megretski2010}.
\end{proof}

\section{SRGs from data trajectories}
\label{sec:trajectories}
From now on, we assume that the underlying description of the operators $\Ttau$ and $\Tinf$ are unknown, but we have an input-output trajectory $(\u_\k,\y_\k)_{k = 0}^{N-1}$ of the system available. In this section, we show that it is possible to compute the maximum and minimum gain of $\lim_{\tau \rightarrow \infty} \Ttau$ and $\Tinf$ based on the trajectories, following the results in \cite{Koch22}. 
This allows us to compute the \gls{srg} using Corollary \ref{cor:circle_thm} from a single input-output trajectory.
We start by assuming noise-free measurements, and then we extend the result to noisy measurements in Section~\ref{sec:noisetrajectories}. 

To be able to use input-output trajectories for the gain computations, we make two assumptions. First, the data trajectory that we use needs to contain enough information. This can be assured by requiring that the input trajectory is sufficiently persistently exciting.
\begin{definition}
    The sequence $(u_k)_{k=0}^{N-1}$ is persistently exciting of order $L$ if the corresponding Hankel matrix
    \begin{equation}
        \begin{pmatrix}u_0 & u_1 & \dots & u_{N-L}\\
        u_1 & u_2 & \dots & u_{N-L+1}\\
        \vdots & \vdots & \ddots & \vdots\\
        u_{L-1} & u_L & \dots & u_{N-1}
        \end{pmatrix}
    \end{equation}
    has rank $mL$.
\end{definition}

Secondly, we assume knowledge of an upper bound on the lag of the observable state-space representation associated with the operators $\Ttau$ and $\Tinf$. The lag is defined as follows.
\begin{definition}
    The lag of a system on state-space form as in~\eqref{eq:system}, is the smallest $l  \in \{1,2,3,\dots\}$ such that the matrix
    \begin{equation}
            \begin{pmatrix}
                 C^\top & (CA)^\top & \cdots & (CA^{l-1})^\top
            \end{pmatrix}^\top          
    \end{equation}
    has rank $n$, where $n$ is the number of states.
\end{definition}
To allow us to perform maximum and minimum gain computations of the operators $\Ttau$ and $\Tinf$ using \glspl{lmi} based on the data, we first define 
\begin{equation}
    \xi_k \! = \! \begin{pmatrix}
        u_{k-l}^\top \! & \! u_{k-l+1}^\top & \!\! \cdots \!\! & u_{k-1}^\top \! & \! y_{k-l}^\top \! & \! y_{k-l+1}^\top &\!\! \cdots \!\!& y_{k-1}^\top
    \end{pmatrix}^\top \!,
    \label{eq:xi}
\end{equation}
for some $l$ greater than or equal to the lag of the state-space representation associated with the
operators. Then we organise the data in the following matrices
\begin{equation}
\setlength\arraycolsep{3pt}
\begin{aligned}
\Xi &\coloneqq \begin{pmatrix} \xi_l & \xi_{l+1} & \cdots & \xi_{N-1}\end{pmatrix}, \,\,
\Xi_{+} \coloneqq \begin{pmatrix} \xi_{l+1} & \xi_{l+2} & \cdots & \xi_{N}\end{pmatrix} \\
U_\Xi &\coloneqq \begin{pmatrix} \u_l & \u_{l+1} & \cdots & \u_{N-1}\end{pmatrix}, \,\,
Y_\Xi \coloneqq \begin{pmatrix} \y_l & \y_{l+1} & \cdots & \y_{N-1}\end{pmatrix}\!.
\end{aligned}
\label{eq:sequences}
\end{equation} 
This leads us to the main result of the section. 
\begin{theorem}
    Consider the operators $\Ttau$ and $\Tinf$ as defined in Definition \ref{def:Ttau} and \ref{def:Tinf}. Let $l$ be greater than or equal to the lag of the associated system in \eqref{eq:system}, and let $(u_k, y_k)_{k=0}^{N-1}$ be an input-output trajectory of the associated system, where $(u_k)_{k=0}^{N-1}$ is persistently exciting of order $n + l + 1$. Then
    \begin{align}    
    & \lim_{\tau \rightarrow \infty} \maxgain{\Ttau} = \inf \gamma \\
    & \text{such that there exists a}\;\; \P = \P^\top \succeq 0 \;\; \text{satisfying}:\\
    &  \Xi_{+}^\top P \Xi_{+} - \Xi^\top P \Xi + Y_\Xi^\top Y_\Xi - \gamma^2 U_\Xi^\top U_\Xi \preceq 0
    \label{eq:LMI_max_extended}
    \end{align}
    and 
    \begin{align}    
    & \lim_{\tau \rightarrow \infty} \mingain{\Ttau} = \sup \zeta \\
    & \text{such that there exists a}\;\; \P = \P^\top \succeq 0 \;\; \text{satisfying}:\\
    &  \Xi_{+}^\top P \Xi_{+} - \Xi^\top P \Xi - Y_\Xi^\top Y_\Xi + \zeta^2 U_\Xi^\top U_\Xi \preceq 0
    \label{eq:LMI_min_extended}
    \end{align}
     Without the constraint $\P \succeq 0$, we get $\maxgain{\Tinf} = \inf \gamma$ and $\mingain{\Tinf} = \sup \zeta$ instead from \eqref{eq:LMI_max_extended} and \eqref{eq:LMI_min_extended}, respectively. 
\label{thm:LMIs_extended}
\end{theorem}
\begin{proof}
    From \cite[Theorems 1 and 5]{Koch22} it follows that the \glspl{lmi} in \eqref{eq:LMI_max_extended} and \eqref{eq:LMI_min_extended} are equivalent to \eqref{eq:LMI_maxgain} and \eqref{eq:LMI_mingain}, respectively, when $l$ is greater than or equal to the lag of the minimal state-space realisation associated with the operators and the input trajectory $(u_k)_{k=0}^{N-1}$ is persistently exciting of order $n + l + 1$. Then it follows from Theorem~\ref{thm:LMIs} that $\lim_{\tau \rightarrow \infty} \maxgain{\Ttau} =~\inf\gamma$ and $\lim_{\tau \rightarrow \infty} \mingain{\Ttau} = \sup \zeta$. 
    
    Similarly, without the constraint $P \succeq 0$, the result follows from combining the results in \cite[Theorems 1 and 5]{Koch22} with the continuous-time counterpart in~\cite[Theorem 2.8]{Scherer2000} and its connection to discrete-time following~\cite[Chapter 4.6]{Scherer2000}. Note that for \cite[Theorem 5]{Koch22} the result can be proved analogously when $P \not\succeq 0$.    
\end{proof}

\begin{remark}
    Note that the maximum and minimum gain of $\Tinf - \alpha \id$ and $\lim_{\tau \rightarrow \infty} \Ttau - \alpha \id$ can be calculated by replacing $\y_\k$ with $\y_\k - \alpha\u_\k$ for all $k$ in the matrices $\Xi, \Xi_{+}$ and $Y_{\Xi}$ in~\eqref{eq:sequences}. The same holds for Theorem \ref{thm:Robust_LMIs} below.
\end{remark} 

\begin{remark}
    An upper bound on the lag of a state-space system that can be used as $l$ in the results is, for example, the number of states in the state-space realisation. 
    Note however that larger $l$ yields larger \glspl{lmi}.
\end{remark} 

\section{Robust SRGs from noisy data trajectories}
\label{sec:noisetrajectories}
We will now consider the case when the data trajectory is corrupted by noise and use Corollary~\ref{cor:circle_thm} to compute a robust version of the \gls{srg} of $\Tinf$ and $\lim_{\tau \rightarrow \infty} \Ttau$. The robust version is guaranteed to cover the actual \gls{srg} given that the noise comes from a known set, as described below.

From \cite[Lemma 3]{Koch22}, it follows that the system in~\eqref{eq:system} can also be written on difference operator form
\begin{equation}
\begin{aligned}
        y_k = & - A_l y_{k-1} - \cdots - A_2 y_{k-l+1} - A_1y_{k-l} \\ &+ Du_k + B_l u_{k-1} + \cdots + B_2 u_{k-l+1} + B_1u_{k-l},
\end{aligned}
    \label{eq:difference_operator}
\end{equation} 
for some $A_i \in \Real^{m \times m}$ and $B_i \in \Real^{m \times m}$ with $i = 1, \dots l$ and $l$ greater than or equal to the lag of~\eqref{eq:system}. As this model yields the same input-output trajectories as the model in \eqref{eq:system}, it can likewise be associated with the operators $\Ttau$ and $\Tinf$ as in Definition~\ref{def:Ttau} and~\ref{def:Tinf}. 
We model the effect of noise on the system with process noise $v_k \in \Real^{m_v}$ following~\cite{Koch22} as below
\begin{equation}
\begin{aligned}
    y_k = &- A_l y_{k-1} - \cdots - A_2 y_{k-l+1} - A_1y_{k-l} \\& + Du_k + B_l u_{k-1} + \cdots + B_1u_{k-l} + B_vv_k,
\end{aligned}
    \label{eq:difference_operator_noise}
\end{equation}
 where $B_v \in \Real^{m \times
 m_v}$. Prior knowledge about how noise affects the system can be modelled with $B_v$. If no such information is available, we can simply pick $B_v = I$. 
 
The noisy input-output behaviour in \eqref{eq:difference_operator_noise} can be represented by a state-space setup with state $\xi$ as in \eqref{eq:xi}, on the following form
\begin{equation}
\setlength\arraycolsep{3.5pt} 
\begin{aligned}
\label{eq:extended_system_matrix}
\begin{pmatrix} u_{k-l+1} \\ \vdots \\ u_{k-1} \\ u_{k} \\ y_{k-l+1} \\ \vdots \\ y_{k-1} \\ y_{k} \end{pmatrix} = \begin{pmatrix} 
0 & I & \dots & 0 & 0 & 0 & \dots & 0  \\
\vdots & \ddots & \ddots&\ddots& \vdots & & \ddots&\vdots \\
0 & 0 & \dots & I & 0& 0&  \dots & 0 \\
0 & 0 & \dots & 0 & 0 & 0&  \dots & 0 \\
0 & 0 & \dots & 0 & 0& I& \dots &0 \\
\vdots & \ddots & \ddots&\vdots&\vdots & & \ddots&\vdots \\
0 & 0 & \dots & 0 & 0& 0&  \dots & I \\
B_1 & B_2 & \dots & B_l & -A_1 & -A_2 & \dots &  -A_l 
\end{pmatrix}  
\begin{pmatrix} u_{k-l} \\ u_{k-l+1} \\ \vdots \\ u_{k-1} \\ y_{k-l} \\ y_{k-l+1} \\ \vdots \\ y_{k-1}  \end{pmatrix} + 
\begin{pmatrix} 
0 \\ \vdots \\ 0 \\ I \\ 0 \\ \vdots \\ 0 \\ D  \end{pmatrix} 
u_{k} + \begin{pmatrix} 0 \\ \vdots \\ 0 \\ 0 \\ 0 \\ \vdots \\ 0 \\  B_v \end{pmatrix} v_{k} 
\end{aligned}
\end{equation}

As we assume that the system is unknown except for an input-output trajectory $(\u_\k,\y_\k)_{k = 0}^{N-1}$ only the first rows in this state-space model are known, while the last row must be estimated from the data trajectories. We can therefore write the state-space model on the following form
\begin{equation}
\begin{aligned}
    \label{eq:system_uncertain}
    \xi_{k+1} &= \begin{pmatrix} \widetilde{A} \\ \widetilde{C}  \end{pmatrix} \xi_k + \begin{pmatrix} \widetilde{B} \\ \widetilde{D}  \end{pmatrix} \u_k + \begin{pmatrix} 0\\ B_v  \end{pmatrix} v_k \\
    y_k & = \widetilde{C} \xi_k + \widetilde{D}u_k + B_vv_k,
\end{aligned}
\end{equation}
where $\widetilde{A} \in \Real^{(2ml-m) \times 2ml}$, $\widetilde{B} \in \Real^{(2ml-m) \times m}$ and $B_v$ are known, while $\widetilde{C} \in \Real^{m \times 2ml}$, $\widetilde{D} \in \Real^{m \times m}$ and $v_k$ are unknown.
If the unknown matrices $\widetilde{C}$ and $\widetilde{D}$ belong to a known set, we can use robust control tools to compute bounds on the maximum and minimum gain for the operators $\Ttau$ and $\Tinf$.
We obtain such a set by assuming a combination of quadratic and linear bounds on the noise sequence.
\begin{assumption}
The matrix $(v_l \quad v_{l+1} \quad \cdots \quad v_{N-1})$
belongs to the set
\begin{equation}
    \mathcal{V} \! \coloneqq \left\{V \in \Real^{m \times(N-l)} \! : \begin{pmatrix}
       V^\top \\ I
    \end{pmatrix}^\top \begin{pmatrix}
        Q & S \\ S^\top & R
    \end{pmatrix}  \begin{pmatrix}
       V^\top \\ I
    \end{pmatrix} \succeq 0 \right\},
\end{equation}
where $Q \in \Real^{(N-l) \times (N-l)}$, $S \in \Real^{(N-l) \times m_v}$, and   $R \in \Real^{m_v \times m_v}$, with $Q$ and $R$ symmetric and in addition $Q \prec 0$.
\label{noise_assumption}
\end{assumption}

This assumption gives a set of possible $\widetilde{C}$ and $\widetilde{D}$ that are consistent with the data trajectory for some noise sequence $V \in \mathcal{V}$, including the $\widetilde{C}$ and $\widetilde{D}$ of the actual model.
Hence, we define the \textit{robust \gls{srg}} for the uncertain model in~\eqref{eq:system_uncertain} to be the set of \glspl{srg} of all models that are consistent with the data and bounds on the noise. 
As the actual model is included in the set, this will over-approximate the actual \gls{srg}.
Define 
    \begin{align}
    \setlength\arraycolsep{3.5pt}
     \begin{pmatrix}
        \overline{Q} & \overline{S} \\ \overline{S}^\top & \overline{R}
    \end{pmatrix} \! \coloneqq & \!
    \setlength\arraycolsep{1pt}
     \begin{pmatrix} \! \begin{pmatrix}
        *  \\ * 
    \end{pmatrix} \! & \! \begin{pmatrix}
        * \\ *
    \end{pmatrix} \! \\  *  \! &  \! * \end{pmatrix} 
    \!\!\begin{pmatrix}
        Q & SB_v^\top \\ S^\top B_v & B_v R B_v^\top
    \end{pmatrix} \!\! \begin{pmatrix} \! \begin{pmatrix}
        \Xi  \\ U_\Xi 
    \end{pmatrix} \! & \! \begin{pmatrix}
        0 \\ 0
    \end{pmatrix} \\  Y_\Xi \! & \! I \end{pmatrix}^{\!\top} \\
    \label{eq:inverse}
    \text{and} \quad  & \begin{pmatrix}
        \widetilde{Q} & \widetilde{S} \\ \widetilde{S}^\top & \widetilde{R}
    \end{pmatrix} \coloneqq      \begin{pmatrix}
        \overline{Q} & \overline{S} \\ \overline{S}^\top & \overline{R}
    \end{pmatrix}^{-1}
\end{align}
assuming that the inverse exists.
If $Q,S,R$, and $B_v$ are chosen such that the inner matrix is invertible, invertibility of \eqref{eq:inverse} follows from the outer matrix having full column rank.
Since the set of noise realisations that do not fulfil this has measure zero under sufficiently persistently exciting inputs, invertibility holds with probability one for continuous noise distributions.
We now show how to compute the robust \gls{srg}.
\begin{theorem}
Let Assumption \ref{noise_assumption} hold and assume the matrix in \eqref{eq:inverse} is invertible. Consider the operators $\Ttau$ and $\Tinf$ as defined in Definition \ref{def:Ttau} and \ref{def:Tinf}, and let $l$ be greater than or equal to the lag of the associated system in~\eqref{eq:system}. If there exists a $P = P^\top \succeq 0$ and $\tau \geq 0$ such that \eqref{eq:noise_LMI} holds for $\Lambda = -\gamma^2I$ and $\Psi = I$, then $$\lim_{\tau \rightarrow \infty} \maxgain{\Ttau} \leq \gamma.$$ Likewise, if there exists a $P = P^\top \succeq 0$ and $\tau \geq 0$ such that \eqref{eq:noise_LMI} holds for $\Lambda = \zeta^2I$ and $\Psi = -I$, then $$\lim_{\tau \rightarrow \infty} \mingain{\Ttau} \geq \zeta.$$ If we drop the requirement that $P \succeq 0$, we instead get that $\maxgain{\Tinf}~\leq~\gamma$ and $\mingain{\Ttau}~\geq~\zeta$. 
\label{thm:Robust_LMIs}
\end{theorem}

\begin{equation}
\setlength\arraycolsep{3.5pt} 
\begin{aligned} \label{eq:noise_LMI}
&\left( \!
\begin{array}{ccc}
I & 0 & 0 \\
\begin{pmatrix}\widetilde{A} \\ 0\end{pmatrix} & \begin{pmatrix}\widetilde{B} \\ 0\end{pmatrix} & \begin{pmatrix}0 \\ I \end{pmatrix} \\ 
0 & I & 0 \\ 
0 & 0 & I  
\end{array}
\! \right)^\top \!\!
\left(
\begin{array}{cccc}
-P&0&0&0\\
0&P&0&0\\
0&0& \Lambda & 0 \\
0&0& 0 & \Psi 
\end{array}\right)
\left( \!
\begin{array}{ccc}
I & 0 & 0 \\
\begin{pmatrix}\widetilde{A} \\ 0\end{pmatrix} & \begin{pmatrix}\widetilde{B} \\ 0\end{pmatrix} & \begin{pmatrix}0 \\ I \end{pmatrix} \\ 
0 & I & 0 \\ 
0 & 0 & I \\ 
\end{array}
\! \right)  -  \tau \begin{pmatrix}
        \widetilde{Q} & \widetilde{S} \\ \widetilde{S}^\top & \widetilde{R}
    \end{pmatrix} \preceq 0 \\
& \text{Maximum gain: } \Lambda = -\gamma^2 I \quad \Psi = I \quad \text{Minimum gain: } \Lambda = \zeta^2 I \quad \Psi = -I
\end{aligned}
\end{equation}

\begin{proof}
    We first show that the set of all possible matrix pairs $(\widetilde{C},\widetilde{D})$ according to the noise model
    \begin{equation}
        \Sigma = \{ (\widetilde{C}, \widetilde{D}) : Y_\Xi = \widetilde{C}\Xi + \widetilde{D}U_\Xi + B_v  V, V\in \mathcal{V} \},
    \end{equation}
    is equal to
    \begin{equation}
    \left\{ (\widetilde{C}, \widetilde{D}) :  \begin{pmatrix}
       I & 0 \\ 0 & I \\ \widetilde{C} & \widetilde{D}
    \end{pmatrix}^\top \!\! \begin{pmatrix}
        \widetilde{Q} & \widetilde{S} \\ \widetilde{S}^\top & \widetilde{R}
    \end{pmatrix} \!  \begin{pmatrix}
       I & 0 \\ 0 & I \\ \widetilde{C} & \widetilde{D}
    \end{pmatrix} \preceq 0 \right\}.
    \label{eq:noise_set}
    \end{equation}
    It follows along the lines of the proof of \cite[Remark 2 and Lemma 4]{vWaarde22} that $\Sigma$ is equal to
    \begin{align}
        \left\{ (\widetilde{C}, \widetilde{D}) : 
        \begin{pmatrix} -\widetilde{C}^\top \\ -\widetilde{D}^\top \\ I \end{pmatrix}^\top \begin{pmatrix} \overline{Q} & \overline{S} \\ \overline{S}^\top & \overline{R} \end{pmatrix}
        \begin{pmatrix} -\widetilde{C}^\top \\ -\widetilde{D}^\top \\ I \end{pmatrix} \succeq 0 \right\}. 
    \end{align}
    As the invertibility of the matrix in \eqref{eq:inverse} ensures that $\begin{pmatrix}
        \Xi^\top & U_{\Xi}^\top 
     \end{pmatrix}$ has full column rank, and we assume $Q \prec 0$, we know that $\overline{Q} \prec 0$. The dualisation lemma \cite[Lemma 4.9]{Scherer2000} then gives that $\Sigma$ is equal to~\eqref{eq:noise_set}.  

    Now we multiply \eqref{eq:noise_LMI} with 
    \begin{equation}
        \begin{pmatrix}
       I & 0 \\ 0 & I \\ \widetilde{C} & \widetilde{D}
    \end{pmatrix}
    \end{equation}
    from the right and its transpose from the left.
    It then follows from the fact that~\eqref{eq:noise_set} is equal to $\Sigma$ that
    the existence of a $P = P^\top \succeq 0$ and $\tau \geq 0$ such that \eqref{eq:noise_LMI} holds implies that 
    there exists a $P = P^\top \succeq 0$ such that 
    \begin{equation}
    \begin{aligned}
    \setlength\arraycolsep{3pt} 
    \left(
    \begin{array}{cccc}
    \star & \star \\ 
    \star & \star \\ 
    \star & \star \\ 
    \star & \star
    \end{array}
    \right)^\top
    \left(
    \begin{array}{cccc}
    -P&0&0&0\\
    0&P&0&0\\
    0&0&\Lambda&0\\
    0&0&0&\Psi
    \end{array}\right)
     \left( \!\!
    \begin{array}{cccc}
    I & 0 \\
   \begin{pmatrix}
       \widetilde{A} \\ \widetilde{C}
   \end{pmatrix} & \begin{pmatrix}
       \widetilde{B} \\ \widetilde{D}
   \end{pmatrix} \\ 
    0 & I \\ 
    \widetilde{C} & \widetilde{D}
    \end{array} \!\!
    \right) \preceq 0
    \end{aligned}
    \end{equation}
    for all possible $(\widetilde{C}, \widetilde{D}) \in \Sigma$. 
    So if there exists a $P = P^\top \succeq 0$ and $\tau \geq 0$ such that \eqref{eq:noise_LMI} holds for $\Lambda = -\gamma^2I$ and $\Psi = I$ this implies that there exists a non-negative storage function $V(\xi) = \xi^\top P\xi$ such that
    $V(\xi_{k+1}) - V(\xi_k) \leq   \gamma^2\norm{u_k}^2 - \norm{y_k}^2$
    for the trajectories of all possible systems \eqref{eq:difference_operator} consistent with set $(\widetilde{C}, \widetilde{D}) \in \Sigma$. It then follows from \cite[Theorem~1 and~2]{Koch22} that feasibility of  \eqref{eq:noise_LMI} with $\Lambda = -\gamma^2I$ and $\Psi = I$ implies 
    feasibility of \eqref{eq:LMI_maxgain}, as the actual system is included in the set $\Sigma$. 
    
  The proof for the lower bound on the minimum gain follows analogously. As well as for the proofs where the constraint $P \succeq 0$ is removed, with the only difference that the storage function no longer has to be non-negative. 
\end{proof}

\section{Examples} \label{sec:examples}
In this section, we show \glspl{srg} of operators $\Tinf$ and $\Ttau$ associated with a discrete-time \gls{lti} system on state-space form, computed with the different methods presented in the paper. The examples illustrate that 
\begin{itemize}
    \item The \gls{srg} obtained with Theorems~\ref{thm:LMIs} and~\ref{thm:LMIs_extended} are the same.
    \item The robust \gls{srg} contains the \gls{srg} of the underlying system and is monotone increasing with noise level.
    \item Systems with the same \gls{srg} can have different robust \glspl{srg}.
\end{itemize}

To ensure that the robust \gls{srg} is the smallest possible set, we minimise $\gamma$ and maximise $\zeta$ for Theorem \ref{thm:Robust_LMIs}. When we generate noisy data trajectories for the computations, we sample the noise uniformly from a ball such that $\norm{v_k}_{2} \leq \Bar{v}$ for all $k = 0, \dots,N-1$. This means that the assumption on the noise from Assumption \ref{noise_assumption} has the parameters $Q = -I$, $S = 0$ and $R = \Bar{v}^2(N-l)I$, where we let  $l$ be the lag of the system in~\eqref{eq:system} associated with the operators. We also assume no prior knowledge of how the noise affects the system, hence $B_v = I$. The \glspl{lmi} are solved using MOSEK 11~\cite{mosek}. For all examples we use $N = 50$, and $1000$ equally spaced $\alpha$ points from $-20$ to $20$. Levels of $\Bar{v}$ are specified in the figures.
\subsection{Unstable MIMO}
In Figure~\ref{fig:unstable_MIMO}, we see an example of the \glspl{srg} for the two different operators $\Tinf$ and $\lim_{\tau \rightarrow \infty} \Ttau$ associated with an unstable \gls{mimo} system, which has state-space representation 
\begin{equation}
    \label{eq:ss_MIMO_ex}
    \begin{aligned}
        \x_{\k+1} &= \begin{pmatrix}
            0.5  & 0 & 0 &0 \\ 0 & 1.05 & 0 & 0 \\ 0 & 0 & -0.3 & 0 \\ 0& 0 & 0 & -0.9
        \end{pmatrix}\x_{\k} + \begin{pmatrix}
            -2 & 0 \\ 1 & 0 \\ 1 & -2 \\ -1 & 0
        \end{pmatrix}\u_{\k}\\
         \y_{\k} &= \begin{pmatrix}
             0.2 &-0.3 & 0.4 &0 \\ 0 &0.1 &-0.3& 0.5
         \end{pmatrix}\x_{\k}.
    \end{aligned}
\end{equation}
This shows that the \gls{srg} obtained from the state-space representation using Theorem~\ref{thm:LMIs} and the \gls{srg} obtained from data trajectories using Theorem~\ref{thm:LMIs_extended} are the same. 

We also see that the robust versions of the \gls{srg} obtained using Theorem~\ref{thm:Robust_LMIs} gives a slightly bigger area that includes the original \gls{srg}. Furthermore, the \gls{srg} for the associated operator $\lim_{\tau \rightarrow \infty} \Ttau$ includes $\{\infty\}$ as the system is unstable.

\begin{figure}
    \centering
      \begin{tikzpicture}[scale = 0.3, >={Stealth[scale=1]}] 
      \begin{scope}[every to/.style={hyperbolic plane}]
      \input{data/SRG/robust_MIMO_unstable.tex}
      \input{data/SRG/trajectory_MIMO_unstable.tex}
      \input{data/SRG/normal_MIMO_unstable.tex}
      \end{scope}

    \node[right] at (3,7) {$\Tinf$};

      \foreach \x in {-3,2,6}{
        \node[below] at (\x,-0.1) {\x};
        \draw (\x,0) -- (\x,-0.2);
      }

    \foreach \y in {-4, 4}{
        \node[left] at (-0.05,\y) {\y};
        \draw (0,\y) -- (-0.2,\y);
    }


      \draw[->] (-4.5,0) -- (7.5,0);
      \draw[->] (0,-7) -- (0,7) node[left] {$\Im(s)$};

      \begin{scope}[shift={(13.5,0)}]
         \fill[DutchOrange!70] (-4.4,-6.9) rectangle (7,6.4);   
      \input{data/SRG/MIMO_int.tex}

      \node[right] at (1,7) {$\lim_{\tau \rightarrow \infty} \Ttau$};

      \foreach \x in {-3,2,6}{
        \node[below] at (\x,-0.1) {\x};
        \draw (\x,0) -- (\x,-0.2);
      }

    \foreach \y in {-4, 4}{
        \node[left] at (-0.05,\y) {\y};
        \draw (0,\y) -- (-0.2,\y);
    }


      \draw[->] (-4.5,0) -- (7.5,0) node[above] {$\Re(s)$};
      \draw[->] (0,-7) -- (0,7) node[left] {$\Im(s)$};
      \node[left] at (6,5) {$\cup \{\infty\}$};

     \node[right] at (2.5, -4.76) {$\Bar{v}: 0.01$};
        \draw[fill=DutchBlue!80] (2, -4.6) rectangle (2.66, -4.98);
      \end{scope}
    \end{tikzpicture}
    \caption{The \glspl{srg} for the operators $\Tinf$, to the left, and  $\lim_{\tau \rightarrow \infty} \Ttau$, to the right, of the system in \eqref{eq:ss_MIMO_ex}. The orange area is the \gls{srg} obtained from state-space representation or data trajectories, while the blue area shows the robust extension obtained from noisy data trajectories.} 
    \label{fig:unstable_MIMO}
\end{figure}

\subsection{Low- and High-pass filters}
The \glspl{srg} and their robust versions with different noise levels for a low-pass filter and a high-pass filter are shown in Figure~\ref{fig:low_high}. The operators we consider are the corresponding $\Tinf$. 
The low-pass filter has state-space representation
\begin{equation}
    \begin{aligned}
        \x_{\k+1} &= \begin{pmatrix}
            0.94  & -0.33 \\ 1 & 0
        \end{pmatrix}\x_{\k} + \begin{pmatrix}
            1 \\ 0
        \end{pmatrix}\u_{\k}\\
         \y_{\k} &= \begin{pmatrix}
             0.29 & 0.07
         \end{pmatrix}\x_{\k} + 0.10 \, \u_{\k},
         \label{eq:lowpass}
    \end{aligned}
\end{equation}
and the high-pass filter has state-space representation
\begin{equation}
    \begin{aligned}
        \x_{\k+1} &= \begin{pmatrix}
            0.94  & -0.33 \\ 1 & 0
        \end{pmatrix}\x_{\k} + \begin{pmatrix}
            1 \\ 0
        \end{pmatrix}\u_{\k}\\
         \y_{\k} &= \begin{pmatrix}
             -0.60 & 0.38
         \end{pmatrix}\x_{\k} + 0.57 \, \u_{\k}.
         \label{eq:highpass}
    \end{aligned}
\end{equation} 
Both these systems have the same \gls{srg}, but the robust \glspl{srg} obtained from noisy data trajectories are different. The Nyquist diagrams of these systems have exactly overlapping curves, yielding the same \glspl{srg}. However, the frequency mapping for the systems are not equivalent, which means that the same point on the Nyquist diagram corresponds to different frequencies for both systems. For both systems the noise model in \eqref{eq:difference_operator_noise}, used for the computation of the robust \gls{srg}, gives the low-pass filter $y_k = 0.94y_{k-1} - 0.33y_{k-2} + v_k$ from noise to output. This means that we expect the noise to affect low frequencies most, explaining the differences in the respective robust \glspl{srg} in Figure~\ref{fig:low_high}.
Figure~\ref{fig:low_high} also shows that the robust \gls{srg} grows monotonically with noise level, forming nested subsets of~$\Compex$.

\begin{figure}
    \centering
        \begin{tikzpicture}[every to/.style={hyperbolic plane}, scale=1.82, >={Stealth[scale=1]}]     
      \input{data/SRG/robust_low_inf_0.02.tex}
      \input{data/SRG/robust_low_inf_0.01.tex}
      \input{data/SRG/robust_low_inf_0.005.tex}
      \input{data/SRG/normal_lowpass_inf.tex}
     \node[right] at (0.1,1.15) {\textbf{Low-pass filter}};

      \foreach \x in {0.5}{
        \node[below] at (\x,-0.01) {\x};
        \draw (\x,0) -- (\x,-0.05);
      }

    \foreach \y in {-1,1}{
        \node[left] at (-0.05,\y) {\y};
        \draw (0,\y) -- (-0.05,\y);
    }


      \draw[->] (-0.35,0) -- (1.3,0); 
      \draw[->] (0,-1.05) -- (0,1.15) node[left] {$\Im(s)$};

        \begin{scope}[shift={(2.6,0)}]
        \input{data/SRG/robust_high_inf_0.02.tex}
      \input{data/SRG/robust_high_inf_0.01.tex}
      \input{data/SRG/robust_high_inf_0.005.tex}
        \input{data/SRG/normal_highpass_inf.tex}

          \node[right] at (0.1,1.15) {\textbf{High-pass filter}};

      \foreach \x in { 0.5}{
        \node[below] at (\x,-0.01) {\x};
        \draw (\x,0) -- (\x,-0.05);
      }

    \foreach \y in {-1,1}{
        \node[left] at (-0.05,\y) {\y};
        \draw (0,\y) -- (-0.05,\y);
    }

      \draw[->] (-0.6,0) -- (1.3,0) node[above] {$\Re(s)$};
      \draw[->] (0,-1.05) -- (0,1.15) node[left] {$\Im(s)$};
     \begin{scope}[shift={(0,-1.5)}]
    \node[right] at (1.1,1.1) {$\Bar{v}:$};
      \draw[fill=DutchBlue!90] (0.95, 0.9) rectangle (1.05, 0.95);
        \node[right] at (1.05, 0.925) {$0.005$};
        \draw[fill=DutchBlue!60] (0.95, 0.73) rectangle (1.05, 0.78);
        \node[right] at (1.05, 0.755) {$0.01$};
        \draw[fill=DutchBlue!30] (0.95, 0.56) rectangle (1.05, 0.61);
        \node[right] at (1.05, 0.585) {$0.02$};
    \end{scope}
    \end{scope}
    \end{tikzpicture}
    \caption{The orange area shows the \gls{srg} of the operator $\Tinf$ associated with the filters in \eqref{eq:lowpass} and \eqref{eq:highpass}, obtained from state-space representation or data trajectories. The extended blue areas show the robust version obtained from noisy data trajectories with different noise levels $\Bar{v}$.}
    \label{fig:low_high}
\end{figure}

\section{Conclusions}
\label{sec:conclusions}
We have demonstrated how to compute the \gls{srg} of discrete-time \gls{lti} systems on state-space form in three different ways. If the state-space representation is known, the \gls{srg} can be obtained from \glspl{lmi} based on the state-space matrices. If the system is unknown, we showed how the \gls{srg} can be computed exclusively from data. We distinguish between noise-free and noisy data trajectories, where the latter gives a robust version of the \gls{srg} that contains the \gls{srg} of the actual system.
Throughout the paper, we specify the results for two types of operators defined over $\ell_{2}$ or truncations of~$\ell_{2}$.

\bibliographystyle{plain}
\bibliography{main}

\end{document}